\documentclass[a4paper,autoref,thm-restate]{lipics-v2021}
\usepackage{mathtools}

\pdfoutput=1 
\hideLIPIcs  

\title{On Approximation Schemes for Stabbing Rectilinear Polygons}

\author{Arindam Khan}{Indian Institute of Science, Bengaluru,
India}{arindamkhan@iisc.ac.in}{0000-0001-7505-1687}{Research partly supported by Pratiksha Trust
Young Investigator Award, Google India Research Award, and SERB Core Research Grant
(CRG/2022/001176) on “Optimization under Intractability and Uncertainty”.}
\author{Aditya Subramanian}{Indian Institute of Science, Bengaluru, India}{}{}{}
\author{Tobias Widmann}{Technical University, Munich, Germany}{}{}{}
\author{Andreas Wiese}{Technical University, Munich, Germany}{}{}{}

\authorrunning{A. Khan, A. Subramanian, A. Wiese, and T. Widmann}

\Copyright{Arindam Khan, Aditya Subramanian, Andreas Wiese, and Tobias Widmann}



\category{}

\relatedversion{} 

\acknowledgements{}

\nolinenumbers 

\EventEditors{John Q. Open and Joan R. Access}
\EventNoEds{2}
\EventLongTitle{42nd Conference on Very Important Topics (CVIT 2016)}
\EventShortTitle{CVIT 2016}
\EventAcronym{CVIT}
\EventYear{2016}
\EventDate{December 24--27, 2016}
\EventLocation{Little Whinging, United Kingdom}
\EventLogo{}
\SeriesVolume{42}
\ArticleNo{23}

\ifdefined\DEBUG
    
    \newcommand{\adi}[1]{\textcolor{green}{#1}}

    \def\rem#1{{\marginpar{\raggedright\scriptsize #1}}}
    \newcommand{\arir}[1]{\rem{\textcolor{red}{$\bullet$ #1}}}
    \newcommand{\adir}[1]{\rem{\textcolor{green}{$\bullet$ #1}}}
    \newcommand{\awr}[1]{\rem{\textcolor{blue}{$\bullet$ #1}}}
    \newcommand{\tobr}[1]{\rem{\textcolor{purple}{$\bullet$ #1}}}
\else
    
    \newcommand{\adi}[1]{#1}

    \newcommand{\arir}[1]{}
    \newcommand{\adir}[1]{}
    \newcommand{\awr}[1]{}
    \newcommand{\tobr}[1]{}
\fi

\newcommand{\Exp}{\mathop{\mathbb{E}}}

\newcommand{\eps}{\varepsilon}

\newcommand{\poly}{\mathrm{poly}}

\newcommand{\R}{\mathbb R}
\newcommand{\N}{\mathbb N}
\newcommand{\Z}{\mathbb Z}

\newcommand{\kshape}{\(k\)-shape}
\newcommand{\kshapes}{\(k\)-shapes}

\newcommand{\stabbing}{\textsc{Stabbing}}

\newcommand{\kstabbing}{\(k\)-\stabbing}

\newcommand*{\eqend}{\text{.}}

\newcommand*{\OPT}{\mathrm{OPT}}
\newcommand*{\SOL}{\mathrm{SOL}}
\newcommand*{\calS}{\mathcal{S}}

\newcommand*{\calK}{\mathcal{K}}
\newcommand*{\calC}{\mathcal{C}}
\newcommand*{\calF}{\mathcal{F}}

\newcommand*{\SOPT}{\calS_{\OPT}}

\newcommand*{\cL}{\mathcal{L}}
\newcommand*{\cR}{\mathcal{R}}
\newcommand*{\Kp}{\mathcal{K}_\beta}
\newcommand*{\Lz}{\mathcal{L}_z}
\newcommand*{\Ez}{\K_{\mathrm{rest}}}

\newcommand*{\SEz}{\calS_{\Ez}}
\newcommand*{\Czs}{\sigma_z(s)}

\global\long\def\N{\mathbb{N}}%
\global\long\def\R{\mathbb{R}}%
\global\long\def\DP{\mathrm{DP}}%
\global\long\def\K{\mathcal{K}}%
\global\long\def\S{\mathcal{S}}%

\newcommand*{\wmin}{w_{\operatorfont min}}
\newcommand*{\wmax}{w_{\operatorfont max}}
\newcommand*{\wtotal}{w_{\operatorfont range}}

\NewDocumentCommand{\Log}{E{_}{{}} E{^}{{}}}{%
    \operatorname{log_{#1}^{#2}}%
}

\begin{document}

\maketitle

\begin{abstract}
We study the problem of stabbing rectilinear polygons,  where we are given $n$ rectilinear polygons in the plane
that we want to stab, i.e., we want to select horizontal line segments such that for each given rectilinear polygon there is
a  line segment that intersects two opposite (parallel) edges of it. Our goal is to find a set of  line segments of minimum total length such that all
polygons are stabbed.
For the special case of rectangles, there is a $O(1)$-approximation algorithm and the problem is $\mathsf{NP}$-hard
[Chan, van Dijk, Fleszar, Spoerhase, and  Wolff, 2018]. Also, the problem admits a QPTAS
[Eisenbrand,  Gallato,  Svensson, and Venzin, 2021] and even a PTAS [Khan, Subramanian, and Wiese, 2022].
However, the approximability for the setting of more general polygons, e.g., L-shapes or T-shapes, is completely open.


In this paper, we characterize the conditions under which the problem admits a $(1+\eps)$-approximation algorithm.
We assume that each input polygon is composed of rectangles that are placed on top of each other,
such that for each pair of adjacent edges between rectangles, one edge contains the other.
We show that if all input polygons satisfy the {\em hourglass condition}, then the problem admits a quasi-polynomial time approximation scheme. In particular, it is thus unlikely that this case is $\mathsf{APX}$-hard.
Furthermore, we show that there exists a PTAS if each input polygon is composed out of rectangles with a bounded range of widths. 
On the other hand, if the input polygons do \emph{not} satisfy these conditions, we prove that the
problem is $\mathsf{APX}$-hard, already if all input polygons have only eight edges. We remark that
all polygons with fewer edges automatically satisfy the hourglass condition. On the other hand, for
arbitrary rectilinear polygons we even show a lower bound of $\Omega(\log n)$ for the possible
approximation ratio,
which implies that the best possible ratio is in $\Theta(\log n)$ since the problem is a special case of \textsc{Set Cover}.
%
\end{abstract}

\section{Introduction}\label{sec:introduction}

The \stabbing{} problem is a geometric setting of the well-studied
\textsc{Set Cover} problem. We are given a set of geometric objects
in the plane. The goal is to compute a set of horizontal line segments
of minimum total length such that each given object $R$ is \textit{stabbed},
i.e., there is a line segment $\ell$ for which $R\setminus\ell$
consists of two connected components. The problem was introduced by
Chan, van Dijk, Fleszar, Spoerhase, and Wolff~\cite{ChanD0SW18}
for the case that each given object is an axis-parallel rectangle.
In particular, they argued that this case models a resource allocation
problem for frequencies. In this application, the $x$-axis models
time and the $y$-axis represents a frequency spectrum. Each given
rectangle represents a request for a time window $[t_{1},t_{2}]$
and a frequency band $[f_{1},f_{2}]$ that needs to be fulfilled.
Each selected segment $[t_{1}',t_{2}']\times\{f'\}$ corresponds to
opening a communication channel $f'$ during a time interval $[t_{1}',t_{2}']$
which then serves each request whose time window is contained in $[t_{1}',t_{2}']$
and for which $f$ is a frequency in its corresponding band $[f_{1},f_{2}]$.
Also, Chan et al. \cite{ChanD0SW18} showed a connection to the \textsc{Generalized Minimum
Manhattan Network} problem.

The first result for the case of rectangles was a polynomial time
$O(1)$-approximation due to Chan et al.~\cite{ChanD0SW18}. Subsequently,
Eisenbrand, Gallato, Svensson, and Venzin improved the approximation
ratio to 8 and provided a QPTAS, i.e., a $(1+\eps)$-approximation
algorithm that runs in quasi-polynomial time \cite{QPTAS_EGSV}. In
particular, this implies that the problem is unlikely to be $\mathsf{APX}$-hard.
After that, Khan, Subramanian, and Wiese presented a polynomial time
approximation scheme (PTAS) for rectangles~\cite{KhanSW22}.

A natural question is the \stabbing~problem for geometric shapes
that are more general than rectangles. We restrict ourselves to rectilinear
polygons since, e.g., a triangle can be stabbed at one of its vertices
at essentially zero cost. A rectilinear polygon can model more general
types of requests in the resource allocation application above, for
example, requests for which the allowed frequency bands depend on
the time window. Also from a theoretical point of view, it is natural
to ask which approximation ratios are possible for more general geometric
objects.

As mentioned above, \stabbing~admits a $(1+\eps)$-approximation
algorithm when all given objects are rectangles~\cite{KhanSW22}.
However, is this also true for slightly more general polygons, e.g.,
that have the shape of an L or a T, or even for arbitrary rectilinear
polygons? If not, under which conditions on the input objects is a
$(1+\eps)$-approximation still possible? Also, given that \stabbing~is
a special case of \textsc{Set Cover}, another natural question is
whether it is strictly easier than this problem.

In this paper, we investigate the question above. We focus on a type
of rectilinear polygons that we call \emph{\kshapes.} Intuitively,
a \kshape{} is formed by $k$ rectangles that are stacked on top
of each other such that for any two consecutive rectangles, the top
edge of the bottom rectangle is contained in the bottom edge of the
top rectangle, or vice versa, see Figure~\ref{fig:kshape}.

\subsection{Our contribution}

In this paper, we characterize under which conditions of the \kshapes{}
in the input the \stabbing~problem admits a $(1+\eps)$-approximation
algorithm in (quasi-)polynomial time, which makes it unlikely that it is $\mathsf{APX}$-hard
in these cases.
We provide two separate conditions for this. %
Also, we prove that if the input objects (slightly) violate these
conditions, then the problem becomes $\mathsf{APX}$-hard. For arbitrary \kshapes{},
we prove even that the problem is as difficult as general \textsc{Set
Cover}, which yields a lower bound of $\Omega(\log n)$ for the possible
approximation ratio.

Our first condition on the input \kshapes{} is the \emph{hourglass condition}. It requires
intuitively that the rectangles of each \kshape{} in the input are stacked like an hourglass
(see Figure~\ref{fig:hourglass} and Definitions~\ref{def:kshape} and \ref{def:hourglass}). Formally, it states that if we
consider the rectangle of each \kshape{} of the smallest width, then the rectangles on top of it are
ordered non-decreasingly by width, and an analogous mirrored ordering holds for the rectangles below
it. For example, L-shapes and T-shapes fulfill this condition.  We prove that this setting admits a
$(1+\eps)$-approximation algorithm for any $\eps>0$ in quasi-polynomial running time, i.e., in time
$n^{(\log n/\eps)^{O(1)}}$. In particular, this makes it unlikely that this case is
$\mathsf{APX}$-hard. Our algorithm generalizes the known QPTAS for the case of rectangles
\cite{QPTAS_EGSV}. However, it is arguably simpler. For example, it does not need an
$O(1)$-approximation algorithm for the problem as a subroutine. Instead, we show that the calls to
this subroutine in \cite{QPTAS_EGSV} can be replaced by suitable guessing steps and by a $O(\log
n)$-approximation algorithm for general \textsc{Set Cover}.

Our algorithm is based on a hierarchical decomposition of the plane into smaller and smaller
rectangular regions. Intuitively, given such a region $R$, we guess all line segments that are
relatively long compared to the width of $R$. Then, we partition $R$ into smaller rectangular
regions inside which we will select only shorter line segments. It can happen that a \kshape{} $K$
contained in $R$ is composed of at least one wide rectangle (of similar width as the guessed long
line segments) and of at least one narrow rectangle. If the guessed long line segments do not stab
$K$, then it is clear that $K$ needs to be stabbed by a short line segment (that we select in one of
the subproblems that we recurse into). Such line segments can stab only the narrow rectangles of
$K$. Therefore, in this case we remove the wide rectangle from $K$ and hence make $K$ smaller. The
hourglass condition ensures that after this removal, the remainder of $K$ still consists of only one
connected component. We crucially need this property in order to ensure that the subproblems of $R$
we recurse into form independent subproblems. This would not be the case if the remainder of $K$
consisted of two connected components such that each of them lies in a different subproblem.

While the hourglass condition is crucial for our algorithm above, it could be that it is not needed
in an alternative algorithmic approach that computes a $(1+\eps)$-approximation for general
\kshapes{}.  However, we prove that this is not the case. We show that our problem is
$\mathsf{APX}$-hard, already if the input consists only of $3$-shapes that do not satisfy the
hourglass condition. On the other hand, note that each 2-shape automatically satisfies the hourglass
condition by definition.

In our proof of this $\mathsf{APX}$-hardness result, we construct 3-shapes that are composed out of
three rectangles whose widths differ a lot.  We prove that the latter is necessary in order to prove
that our problem is $\mathsf{APX}$-hard. To this end, we show that it admits a polynomial time
$(1+\eps)$-approximation algorithm for any $\eps>0$ if each \kshape{} is composed out of rectangles
whose widths are in a \emph{constant }range. This yields our second condition under which our
problem admits a $(1+\eps)$-approximation. Our algorithm is a generalization of the PTAS for
rectangles in~\cite{KhanSW22} . One crucial insight is that if the widths of the rectangles of each
input \kshape{} differ by at most a constant factor of $1/\delta$, then we can reduce our problem to
the setting of rectangles by losing only a factor of $O(1/\delta)$.  To do this, we simply replace
each \kshape{} $K$ by the smallest rectangle that contains $K$. We use this insight in one step of
our algorithm where we need a $O(1)$-approximation algorithm as a black-box.  More precisely, we
again partition the input plane hierarchically into smaller and smaller rectangular regions. In the
process, we repeatedly need to compute constant factor approximations for certain sets of \kshapes{}
that intuitively admit a solution whose cost is at most $O(\delta\eps \OPT)$; for those, we use the
mentioned algorithm. We stab all other \kshapes{} with segments whose total cost is at most
$(1+\eps)\OPT$, which yields a PTAS.

We round up our results by showing that for general \kshapes{} and, more generally, even arbitrary
rectilinear polygons that are composed of $k$ rectangles each, \stabbing{}~admits a polynomial time
$O(k)$-approximation algorithm.  A natural question is whether the dependence on $k$ (and the input
size) in the approximation ratio can be avoided and there is, e.g., also a $O(1)$-approximation. We
show that this is not the case: for arbitrary $k$, we prove that \stabbing{} for \kshapes{} is as
difficult as arbitrary instance of \textsc{Set Cover}, which yields a lower bound of $\Omega(\log
n)$ for our approximation ratio.

\subsection{Other related work}
As mentioned above, the \stabbing~problem is a special case of \textsc{Set Cover} which is
$\mathsf{NP}$-hard~\cite{lewis1983michael} and which does not admit a $(c\cdot\ln n)$-approximation
algorithm for \textsc{Set Cover} for any $c<1$, assuming that $\mathsf{P}\ne
\mathsf{NP}$~\cite{dinur2014analytical} (see also \cite{Feige98}).
On the other hand, a simple polynomial time greedy
algorithm~\cite {Chvatal79} achieves an approximation ratio of $O(\log n)$.


Das, Fleszar, Kobourov, Spoerhase, Veeramoni, and Wolff~\cite{das2012polylogarithmic} studied approximation algorithms for the \textsc{Generalized Minimum Manhattan Network (GMMN)} problem,  where  given a set of $n$ pairs of  terminal vertices, the goal is to find a minimum-length rectilinear
tree such that each pair is connected by a Manhattan path.
The currently best known approximation ratio for this problem  is  $(4 +\eps) \log n$, due to Khan,
Subramanian, and Wiese~\cite{KhanSW22} by using  their PTAS for \stabbing~as a subroutine in a variant of the algorithm of Das et al.~\cite{das2012polylogarithmic}.

Gaur, Ibaraki, and Krishnamurti~\cite{gaur2002constant} studied the problem of stabbing rectangles by a minimum number of axis-aligned
lines and obtained an LP-based 2-approximation algorithm. Kovaleva and Spieksma~\cite{kovaleva2006approximation} studied a
weighted generalization of this problem and gave an $O(1)$-approximation algorithm.

Geometric set cover is a related geometric special case of general \textsc{Set Cover}, where the given sets are geometric objects. Br{\"{o}}nnimann and Goodrich \cite{BronnimannG95} first gave an $O(d \log (d\cdot \OPT))$-approximation algorithm for  unweighted geometric set cover where $d$ is the dual VC-dimension of the set system and and $\OPT$  is the value of the optimal solution.
Aronov, Ezra, and Sharir \cite{aronov2010small}  utilized $\eps$-nets to design an $O(\log \log \OPT)$-approximation algorithm for the hitting set problem involving axis-parallel rectangles.
Varadarajan \cite{Varadarajan10} 
providing an improved  approximation algorithm for weighted geometric set cover for fat triangles or disks, and his techniques were
extended by Chan, Grant, K\"{o}nemann, and Sharpe~\cite{chan2012weighted} to any set system with low shallow cell complexity.
Subsequently, Chan and Grant \cite{ChanG14}, and Mustafa, Raman, and Ray \cite{MustafaRR14} have settled
the $\mathsf{APX}$-hardness statuses of (almost) all natural variants for this problem.
Recently, these problem are studied under online and dynamic setting as well \cite{AgarwalCSXX22, ChanHSX22, KhanLRSW23}.

Maximum Independent Set of Rectangles is another related problem.
The problems admits a QPTAS \cite{AdamaszekHW19}, and recently a breakthrough $O(1)$-approximation algorithm was given by Mitchell \cite{mitchell21}. Subsequently, a $(2+\eps)$-approximation guarantee \cite{Galvez22} was achieved.

Rectangle packing and covering problems such as two-dimensional knapsack \cite{Jansen0LS22, KMSW21, GalvezGHI0W17}, two-dimensional bin packing \cite{BansalK14, KhanS21},  strip packing \cite{harren20145, KhanMSW22} etc.~are well-studied in computational geometry and approximation algorithms.
We refer the readers to  \cite{ChristensenKPT17} for a survey on the approximation/online algorithms related to rectangles.

Rectilinear polygons
appear naturally in the context of  circuit design \cite{lienig2020fundamentals}, architectural design \cite{pottmann2007architectural}, geometric information systems \cite{chang2008introduction}, computer graphics \cite{shirley2009fundamentals},  etc.
In computational geometry, often problems (for general polygons) are studied in the rectilinear setting, e.g., the art gallery problem \cite{worman2007polygon},
rectilinear convex hull \cite{ottmann1984definition}, and rectilinear steinter tree \cite{garey1977rectilinear}.

\section{Preliminaries}\label{sec:preliminaries}

We start with some basic definitions and notation.
We represent a given axis-aligned rectangle $R_i$ as the
Cartesian product of two given closed and bounded intervals, i.e., $R_i = [x_i^\ell, x_i^r] \times [y_i^b, y_i^t]$
for given coordinates \(x_i^\ell, x_i^r, y_i^b, y_i^t \in \N\), where \(x_i^\ell \leq x_i^r\) and \(y_i^b
\leq y_i^t\). The following notation will be useful: we define
\begin{itemize}
    \item \(b(R_i) := [x_i^\ell, x_i^r] \times \{y_i^b\}\) as the \textit{bottom edge} of \(R_i\),
    \item \(t(R_i) := [x_i^\ell, x_i^r] \times \{y_i^t\}\) as the \textit{top edge} of \(R_i\), and
    \item \(w(R_i) := (x_i^r - x_i^\ell)\) as the \textit{width} of \(R_i\).
\end{itemize}
A \emph{horizontal line segment} \(s \subset \R^2\) is a Cartesian product
$s = [x^\ell, x^r] \times \{y\}$
with coordinates \(x^\ell, x^r, y \in \N\) and \(x^\ell \leq x^r\).
    We say that \(s\) \textit{stabs} the rectangle \(R_i\) if and only if \(R_i \cap s = x_i^\ell, x_i^r \times \{y\}\).
    Also, we define \(|s| := x^r - x^\ell\) is the \textit{length} or the \textit{cost} of \(s\).
    We will study the \stabbing{} problem in the setting where each given object is a \kshape{}.

\begin{definition}[\(k\)-shape] \label{def:kshape}
    Let \(k \in \N\). A \emph{\kshape{}} $K$ is the union of at most \(k\) axis-aligned rectangles \(R_1 \cup \dots \cup R_k =
    K\) such that  \(t(R_i) \subseteq b(R_{i + 1})\) or \(t(R_i) \supseteq b(R_{i + 1})\) for each \(i \in \{1, \dots, k - 1\}\).
%
\end{definition}


\begin{figure}[ht]
    \centering
        \begin{subfigure}{0.3\textwidth}
            \centering
            \includegraphics[page=1]{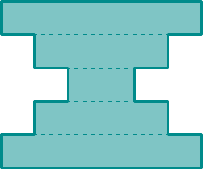}
        \end{subfigure}
        \begin{subfigure}{0.3\textwidth}
            \centering
            \includegraphics[page=2]{fig_kshapes}
        \end{subfigure}
        \begin{subfigure}{0.3\textwidth}
            \centering
            \includegraphics[page=3]{fig_kshapes}
        \end{subfigure}
    \caption{Examples of \kshapes{} satisfying the hourglass condition}
    \label{fig:hourglass}
\end{figure}

\begin{figure}[ht]
    \centering
       \begin{subfigure}{0.45\textwidth}
           \centering
           \includegraphics[page=4]{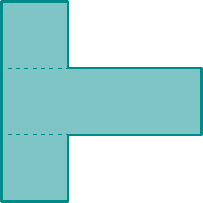}
       \end{subfigure}
       \begin{subfigure}{0.45\textwidth}
           \centering
           \includegraphics[page=5]{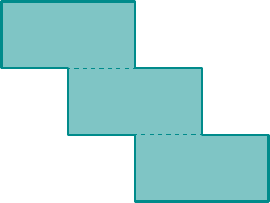}
       \end{subfigure}
    \caption{A 3-shape not satisfying the hourglass condition (left), and a stack of rectangles that does not form a \kshape{} (right).}
    \label{fig:kshape}
\end{figure}

We say that a \kshape{} \(K = R_1 \cup \dots \cup R_k\) is \emph{stabbed} by a line segment \(s\), if there exists an index \(i
\in \{1, \dots, k\}\) such that the rectangle \(R_i\) is stabbed by \(s\). This leads to the
following formal definition of the \stabbing{} problem for \kshape{}s.

\begin{definition}
    Let \(k \in \N\). An instance of the \emph{\kstabbing{} problem for \kshape{}s} is a finite set of \kshapes{}
    \(\calK\), where the objective is to find a set \(\calS\) of horizontal line segments of minimum
    total length, such that every \kshape{} in \(\calK\) is stabbed by a segment in
    \(\calS\).
\end{definition}

In the following section, we shall use the term $\OPT$ interchangeably to refer to the
optimal solution to the problem, and also to represent its cost, i.e., the total length of
segments in the set. Similarly, $\SOL$ will be used to represent a solution set and also its cost.

\section{Quasi-polynomial-time approximation scheme}\label{sec:qptas}

In this section, we present our QPTAS for $k$-\stabbing{}. The algorithm
is an adaptation of the QPTAS for \stabbing{}~\cite{QPTAS_EGSV}
to the more general case of \kshapes{}.

Let $\eps>0$ and suppose we are given a set of \kshape{}s $\K$. In this section, we assume that each
given \kshape{} $K\in \K$ satisfies the hourglass condition (see Figure \ref{fig:hourglass}).

\begin{definition} \label{def:hourglass}
A \kshape{} \(K = R_1 \cup \dots \cup R_k\) satisfies the \emph{hourglass condition} if
there is no value \(i \in \{2, \dots, k - 1\}\) such that both \(w(R_{i - 1}) < w(R_i)\) and \(w(R_{i + 1}) < w(R_i)\).
\end{definition}

For each given \kshape{} $K$, we define
$\wmax(K)\coloneqq\max_{i\in\{1,\dots,k\}}w(R_{i})$ and
$\wmin(K)\coloneqq\min_{i\in\{1,\dots,k\}}w(R_{i})$ which are the widths of the widest and most
narrow parts of $K$, respectively.  For sets of \kshapes{} $\calK'\subseteq\K$, we define
accordingly $\wmax(\calK')\coloneqq\max_{K\in\calK'}\wmax(K)$,
$\wmin(\K')\coloneqq\min_{K\in\calK'}\wmin(K)$. Moreover, we define
$\wtotal(\K')\coloneqq\min\{w | \exists x \forall K\in\K' : K\subseteq[x,x+w]\times\R\}$
as the width of the most narrow strip that contains all
\kshapes{} in $\K'$.
Further, we note here that there are $n$ given \kshapes{} and each is described by at most $2k$
    distinct points. Therefore, the solution to the instance has only $\binom{2kn}{2}$ combinatorially distinct
    candidate segments, which is a polynomial in $n$ (we shall use the notation that the number of
    candidate segments is $\poly(n)$).

\begin{restatable}{lemma}{qptaspreprocess} \label{lem:preprocessing}
    By losing a factor of $1+\eps$ in our approximation ratio, we assume that
    $\frac{\eps}{n}<\wmin(\calK)\leq\wmax(\calK)\leq\log n$ and $\wtotal(\K)\le n\log n$.
\end{restatable}


Let $\mu:=\eps/\log^{2}n$.
We partition the plane into relatively wide
vertical strips of width
$\wmax(\K)/\mu$ each. We do this such that, intuitively, almost all input shapes
are contained
in one of our strips, and the remaining shapes,
which are intersected by the vertical grid lines, can be
stabbed very cheaply. To construct this partition, we define vertical grid lines
with a spacing of $\wmax(\K)/\mu$ and give them a random horizontal shift (see
Figure~\ref{fig:partitioning1}). Then, each shape in $\K$ intersects one of these
grid lines only with very small probability. Therefore, we can show
that there exists a specific way to perform the shift of our grid
lines such that all input shapes intersecting our grid lines can be
stabbed with line segments whose cost is at most $\mu\cdot\OPT$.

\begin{figure}[ht]
    \centering
    \includegraphics{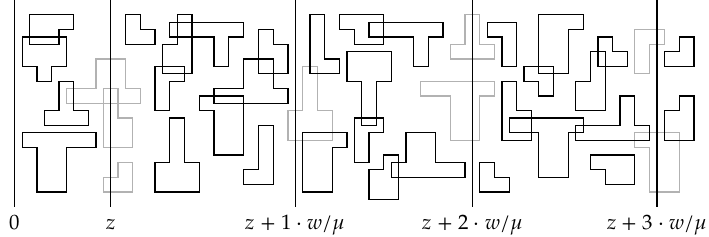}
    \caption{Partitioning the instance into narrow strips.\label{fig:partitioning1}}
\end{figure}
Formally, we invoke the following lemma with our choice for $\mu$ defined above.
It guesses a set of line segments that yield our desired partition
into narrow strips, i.e., it produces a polynomial number of candidate
sets such that one of the has the claimed property. Algorithmically,
we recurse on each of these polynomially many options and at the end
output the returned solution with the smallest total cost.

\begin{restatable}[Partitioning into narrow strips]{lemma}{strippartition} \label{lem:partition1}
Let $\mu>0$ such that $\mu/n<\wmin(\K)$. In polynomial time, we can guess a partition of $\calK$
into sets $\calK_{0},\dots,\calK_{t}$ and one special set $\K_{\mathrm{rest}}$ such that
\begin{romanenumerate}
\item $\OPT\ge\sum_{\ell=1}^{t}\OPT(\K_{\ell})$,
\item $\OPT(\K_{\mathrm{rest}})\le 8\mu\cdot \OPT$, and
\item $\wtotal(\calK_{i})\leq\wmax(\K)/\mu$ for each $i\in\{1,\dots,t\}$.
\end{romanenumerate}
\end{restatable}

We compute an $O(\log n)$-approximate solution for stabbing $\K_{\mathrm{rest}}$
by reducing our problem to an instance
of \textsc{Set Cover} (see Appendix~\ref{app:screduction} for details).
By our choice of $\mu$, the resulting cost is at most $O(\log n\cdot\mu\cdot \OPT)
=O(\OPT\cdot\eps/\log n)$.

Now let $\calK_{i}$ be one of the sets of \kshapes{} due to Lemma~\ref{lem:partition1}.
We define $S_{i}:=[a,b]\times\R$ for some values $a,b\in\R$ with
$b-a\le\wmax(\K)/\mu$ such that each \kshape{} in $\calK_{i}$ is
contained in $S_{i}$. We want to partition $S_{i}$ along horizontal
lines into rectangular pieces such that each resulting piece contains
line segments from $\OPT(\K_{i})$ of total cost at most $O(\wmax(\K)/\mu^{2})$.
To this end, we guess whether the segments in $\OPT(\K_{i})$ have
a total cost of at most $\wmax(\K)/\mu^{2}$. If this is not the case,
we guess a line segment $s=[a,b]\times\{h\}$ for some value $h\in \N$
according to the following lemma, which intuitively partitions $S_{i}$
in a balanced way according to the segments in $\OPT(\K_{i})$. We
call such a segment $s$ a \emph{balanced horizontal cut}.

\begin{restatable}{lemma}{blockpartition} \label{lem:partition2}
If $\OPT(\K_{i})>\wmax(\K)/\mu^{2}$ then
in polynomial time we can guess a value $h\in \N$
and a corresponding line segment $s=[a,b]\times\{h\}$ such that each connected component $C$ of
$S_{i}\setminus s$ contains segments from $\OPT(\calK_{i})$ whose total cost is at least
$\OPT(\K_{i})/2-\wmax(\K)/\mu$.
\end{restatable}

We add $s$ to our solution and recurse on each connected component
$C$ of $S_{i}\setminus s$ separately. The resulting subproblem is
to stab all input shapes that are contained in $C$. Observe that
$s$ stabs all \kshapes{} contained in $S_i$ that intersect both
connected components of $S_i\setminus s$. Given~$C$, we guess
again whether $\OPT(C)$, i.e., the optimal solution for all \kshape{}s contained in $C$,
has a total cost of at most $\wmax(\K)/\mu^{2}$,
and if not, we guess a corresponding horizontal line segment. Note
that we stop after at most $O(\log n)$ recursion levels if all guesses are correct, since
$\OPT(S_i)\le\OPT\le n\log n$
due to our preprocessing in Lemma~\ref{lem:preprocessing}. We enforce that in any case we stop after
$O(\log n)$ recursion levels in order to guarantee a quasi-polynomial bound on the running time later.

\begin{restatable}{lemma}{horizontal}\label{lem:horizontal}
    If all guesses for the balanced horizontal cut are correct, then their total cost
    is bounded by $3\mu\cdot\OPT(\K_{i})$.
\end{restatable}

At the end, each resulting subproblem is characterized by a rectangle $C$ of width at most
$\wmax(\K)/\mu$ and for which $\OPT(C)\le\wmax(\K)/\mu^{2}$.  We guess all line segments in
$\OPT(C)$ whose width is larger than $\eps\wmax(\K)$. Since $\OPT(C)\le\wmax(\K)/\mu^{2}$ there can
be at most $1/\eps\mu^2=\eps^{-3}\log^2n$ of them, and for each of them there are only $\poly(n)$
options. Hence, we can guess them in time $n^{O(\eps^{-3}\log^2n)}$.  Let $\S_{C}$ denote the
guessed segments.

Our next step crucially deviates from the known (Q)PTASs for stabbing
rectangles~\cite{QPTAS_EGSV,KhanSW22}. Inside $C$, there might be a \kshape{} $K$ that is not
stabbed by any segment in $\S_{C}$ but for which one of its rectangles $R_{i}$ satisfies that
$w(R_{i})>\eps\wmax(\K)$.  Since we have guessed all segments in $C$ of width larger than
$\eps\wmax(\K)$ and did not yet stab $K$, we know that the optimal solution does not stab $K$ by
stabbing $R_{i}$ (but by stabbing another rectangle that $K$ is composed of). Therefore, we modify
$K$ by removing $R_{i}$ from $K$. We do this for each rectangle $R_{i}$ with
$w(R_{j})>\eps\wmax(\K)$ that is part of a \kshape{} $K$ that is contained in $C$ but not yet
stabbed. Denote by $\K'(C)$ the resulting set of \kshapes{}.  Importantly, the hourglass property
implies that still each $K$ shapes has only one single connected component. This is the reason why
we imposed this property.

Observe that for each $K\in\K'(C)$ we have that $w_{\max}(K)\le\eps\cdot w_{\max}(\K)$.
Thus, we made progress in the sense that the maximum width of any
\kshape{} reduces by a factor of $\eps$. Also, if all our guesses
are correct, then our total cost is small, i.e., $O(\mu\cdot\OPT)$. Also, the number of guesses
is quasi-polynomially bounded since for each guess there are only
$n^{O(\eps^{-3}\log^2n)}$ many options and our recursion depth is only $O(\log n)$.

\begin{restatable}{lemma}{corrguess} \label{lem:corrguess}
If all our guesses are correct, then the total cost for the selected line segments
due to Lemmas~\ref{lem:partition1} and Lemmas~\ref{lem:horizontal}
is bounded by
$O(\mu\cdot\OPT)$.
Also, the total number of (combinations of) guesses is bounded by
$n^{O(\eps^{-3}\log^2n)}$.
\end{restatable}

We continue recursively with each resulting subproblem. Since initially
$\frac{\eps}{n}<\wmin(\calK)\leq\wmax(\calK)\leq \log n$, we stop after
applying the algorithm above for $O(\log (n/\eps))$ levels. Each
level incurs in total at most $n^{O(\eps^{-3}\log^2n)}$ guesses, which yields
a total running time of $n^{O(\eps^{-4}\log^3n)}$. Also, our approximation ratio can easily
be bounded by $(1+\mu)^{O(\log n)}=1+O(\eps)$.

\begin{theorem}There is a QPTAS for the stabbing problem for \kshapes{}
that satisfy the hourglass condition.
\end{theorem}

\section{PTAS if pieces have bounded ratio of widths} \label{sec:ptas}

In this section, we improve our QPTAS from Section~\ref{sec:qptas} to a PTAS in the special case
that for each given \kshape{}, for any two of its rectangles $R_{i},R_{j}$, it holds that $\delta
w(R_j)\le w(R_{i})\le w(R_{j})/\delta$ for a given constant $\delta>0$.

Let $\alpha$ be a constant for which the problem admits an $\alpha$-approximation algorithm
(We show the existence of such an algorithm in Section~\ref{app:deltaapprox}). Without loss of
generality, we assume that $\alpha,(1/\eps)\in\N$, and we say that an $x$-coordinate $x\in\R$ is
discrete if $x$ is an integral multiple of $\eps^d$, where we define $d\in\N$ such that
$\eps^3/n<\eps^d\le\eps^2/n$; note that hence $d$ is unique. Similarly a $y$-coordinate is called
\textit{discrete} if it is integral. A point is called \textit{discrete} if
its $x$ and $y$ coordinates are discrete, and similarly a segment or a rectangle is said to be
\textit{discrete} if both of its end points, or both of its diagonally opposite corners are
discrete.

\begin{restatable}{lemma}{discretize} \label{lem:discretize}
    Let $\alpha$ be a constant for which \kstabbing{} admits an
    $\alpha$-approximate algorithm and let $\eps>0$ with $\eps<1/3$.
    In polynomial time we can compute a new instance of
    \kstabbing{}, in which each  $K\in\calK$ satisfies,
    \begin{romanenumerate}
        \item $\frac{\alpha\eps}{n}<\wmin(K)\le\wmax(K)\le \alpha$,
        \item all points defining $K$ are discrete,
        \item $K$ lies within a bounding box of $[0,\alpha n]\times[0,(k+1)n]$,
    \end{romanenumerate}
    and this new instance admits a solution of cost at most $(1+O(\eps))\cdot\OPT$ with each
    segment in the solution being discrete, and having length at most $\alpha/\eps$.
\end{restatable}
First, we apply Lemma~\ref{lem:discretize} in order to preprocess our instance.  In our algorithm, we intuitively
embed the recursion of our QPTAS in Section~\ref{sec:qptas} into a polynomial time dynamic program. The idea is
to construct a DP-table that contains one cell for each possible subproblem of a recursive call.
Formally, we introduce one DP-cell $\DP(R,\S)$ for each combination of
\begin{itemize}
    \item a closed rectangle $R\subseteq[0,\alpha n]\times[0,(k+1)n]$ with discrete coordinates,
    \item a set $\S$ of at most $\eps^{-3}$ discrete horizontal line segments,
        that all intersect $R$.
\end{itemize}
This DP-cell encodes the subproblem of stabbing all input \kshapes{} that are
contained in $R$ and that are not already stabbed by the segments in $\S$. Clearly, the DP-cell
$\DP([0,\alpha n]\times[0,(k+1)n],\emptyset)$ corresponds to our given problem.

Given a DP-cell $\DP(R,\S)$, we compute its solution as follows.  The base case occurs when the line
segments in $\S$ already stab all \kshapes{} that are contained in $R$. Then we define
$\DP(R,\S):=\emptyset$.  Another easy case occurs when there is a line segment $\ell\in\S$ that
stabs the interior of $R$ , i.e., $R\setminus\ell$ has two connected components $R_{1}$ and $R_{2}$.
Assume that $\S_{1}$ and $\S_{2}$ are parts of the line segments from $\S$ that intersect $R_{1}$
and $R_{2}$, respectively. Then we define
$\DP(R,\S):=\DP(R_{1},\S_{1})\cup\DP(R_{2},\S_{2})\cup\{\ell\}$.  We will refer to this later as the
\emph{trivial operation.}

Otherwise, we compute a polynomial number of candidate solutions as follows,
\begin{enumerate}
    \item \emph{Add operation. }For each set $\S'$ of discrete segments contained in $R$ for which
        $|\S|+|\S'|\le3\eps^{-3}$ holds, we generate the candidate solution
        $\S'\cup\DP(R,\S\cup\S')$.
    \item \emph{Line operation. }Consider each vertical line $\ell$ that intersects the interior of
        $R$. Let $\K_{\ell}$ denote the set of \kshapes{} contained in $R$ that intersect with
        $\ell$. For each $K\in\K_{\ell}$ we construct the smallest axis-parallel rectangle that
        contains $K$, let $\mathcal{R}_{\ell}$ denote the resulting set of rectangles.  We apply the
        PTAS for stabbing rectangles \cite{KhanSW22} to $\mathcal{R}_{\ell}$, let $\S_{\ell}$ denote
        the computed set of segments. We will show later that the optimal solution for
        $\mathcal{R}_{\ell}$ is by at most a factor $O(1/\delta)$ more expensive that the optimal
        solution for $\K_{\ell}$, and that this approximation ratio is good enough for our purposes
        in this step. Denote by $R_{1}$ and $R_{2}$ the connected components of $R\setminus\ell$ and
        by $\S_{1}$ and $\S_{2}$ the parts of segments from $\S$ that intersect $R_{1}$ and $R_{2}$,
        respectively. We define the candidate solution
        $\S_{\ell}\cup\DP(R_{1},\S_{1})\cup\DP(R_{2},\S_{2})$.
\end{enumerate}
We store in $\DP(R,\S)$ the candidate solution with smallest cost.  Finally, we output the solution
stored in the cell $\DP([0,\alpha n]\times[0,2kn],\emptyset)$.

\subsection{Analysis}
We first note that all DP subproblems and operations are defined on discrete coordinates, and since
there are only a polynomial $\frac{\alpha n}{\eps^d} \times 2kn \le
2\alpha k\eps^{-3}n^3$
number of discrete points, the running time of the dynamic program is also polynomial.

\begin{restatable}{lemma}{ptasruntime} \label{lem:ptasruntime}
    The running time of the above dynamic program is $(kn/\eps)^{O(1/\eps^3)}$.
\end{restatable}

Our consideration of the approximation factor is similar to the analysis of the PTAS by Khan, Subramanian, and
Wiese~\cite{KhanSW22} and our QPTAS in Section~\ref{sec:qptas}. We describe here its main structure
and highlight the key differences. We refer to Appendix~\ref{app:ptas}  for all details.

The solution computed by our DP corresponds to performing a sequence of trivial, add, and
\textit{line} operations, and recursing on the respective subproblems. It is sufficient to argue
that there exists a sequence of these operations such that
\begin{itemize}
    \item there exists a DP-cell for each arising subproblem; in particular, the number of line
        segments in each subproblem is bounded by $3\eps^{-3}$ and,
    \item the total cost of the computed solution is bounded by $(1+O(\eps))\OPT$.
\end{itemize}

We now describe this sequence. It is based on a hierarchical grid of vertices lines, shifted by a
random offset $r\in\{0,\eps^d,2\eps^d,\ldots,\alpha\eps^{-2}\}$ that we will fix later. For each
level $j\in\N_{0}$, we define a grid line $\{r+t\cdot\alpha\eps^{j-2}\}\times\R$ for each
$t\in\Z$. Note that for all $j\le d+2$, grid lines of level  $j$ have discrete $x$-coordinates. We
say that a line segment $\ell\in\OPT$ is of \emph{level }$j$ if the length of $\ell$ is in
$(\alpha\eps^{j},\alpha\eps^{j-1}]$.  We say that a line segment of some level $j$ is
\emph{well-aligned} if its left and right endpoint lies on a grid line of level $j+3$, and if the
$y$-coordinate of both endpoints is discrete.  We can extend each line segment $\ell\in\OPT$ so that
it becomes well-aligned, by increasing its length by at most a factor of $1+O(\eps)$.

\begin{restatable}{lemma}{alignedopt} \label{lem:algn}
    For any value of our offset, by losing a factor of $1+O(\eps)$ in our approximation ratio,
    we can assume that each line segment $\ell\in\OPT$ is well-aligned.
\end{restatable}

Note that each horizontal segment $\ell\in\OPT$ satisfies that
$\alpha\eps/n<|\ell|\le\alpha\eps^{-1}$. By our choice of $d$ we have $\eps^
{d-1}\le\eps/n<\eps^{d-2}$ which implies $\alpha\eps^{d-1}<|\ell|\le\alpha\eps^{-1}$. Since a
segment is of level $j$ if its length is in the range $(\alpha\eps^j,\alpha\eps^{j-1}]$, we can
conclude that all segments in $\OPT$ belong to levels in the range $\{0,\ldots,d-1\}$. From this we
can infer that any well-aligned horizontal segment is aligned to a vertical grid line of level at
most $d+2$, which as we noted earlier has discrete $x$-coordinates.

In our sequence of operations, we first perform one \textit{line operation} for each (vertical) grid line of
level $j=0$. This is similar as partitioning the instance into narrow strips as we did it in
Lemma~\ref{lem:partition1}. However, now each strip has a width of $\alpha\eps^{-2}$ instead of
$w_{\max}(\K)/\mu$. In our following operations, we add horizontal line segments to partition each
vertical strip, similar to Section~\ref{sec:qptas}. Formally, we sort the segments from $\OPT$
of level $j=0$ in increasing order of their $y$-coordinates, and pick every $(\eps^{-3})$-th segment,
and do an add operation along the (strip wide) line along it. This leads to a trivial operation
immediately after that. Finally, we perform add operations for all line segments of level $j=0$ in
$\OPT$. We call the above operations to be \textit{operations of level $0$}.

With the above operations for level $j=0$ done, in increasing order of level $j=1,2,\ldots$ we
do \textit{operations of level $j$} similarly as follows:
\begin{itemize}
    \item \textit{line operations} on vertical grid lines of level $j$,
    \item any valid trivial operations (this step is not done for level $0$),
    \item add, and trivial operations to divide the vertical strips into smaller subproblems,
    \item and finally the add operations on the segments from $\OPT$ of level $j$,
\end{itemize}
mimicking the recursive structure from the analysis of the QPTAS.

\begin{lemma} \label{lem:validop}
    The above sequence of operations always leads to valid DP subproblems.
\end{lemma}
\begin{proof}
 Consider a subproblem $(R,S)$ obtained at any stage of application of the above
 operations. The rectangular cell $R$ is always discrete and a subset of $[0,\alpha n]\times[0,
 (k+1)n]$ since the line and trivial operations are done only on discrete lines. So the only
 property we need to show is that $|S|\le3\eps^{-3}$. Let the last segment added to $S$ be of level
 $j$.  An add operation of level $j$ is preceded by \textit{line operations} of level $j$, and hence any
 segment of level $j-3$ already in $S$ gets removed from $S$ by trivial operations, by virtue of it
 being well-aligned (Lemma~\ref{lem:discretize}). Therefore, $S$ only contains segments from the
 levels $j, j-1$, and $j-2$. By construction, when we perform add operations, we add at most $\eps^
 {-3}$ segments of any particular level $j$ to $S$, and hence there are at most $3\eps^
 {-3}$ segments in $S$.
\end{proof}

We wish to bound the cost of the the above operations. Suppose that we perform a line operation with
a vertical line $\ell$ and let $\K_{\ell}$ denote the \kshapes{} that $\ell$ intersects. Recall that
for each \textit{line operation}, we compute a solution that stabs all \kshapes{} in $\K_{\ell}$ (and in fact
every rectangle in $\mathcal{R}_\ell$). Note that any horizontal line segment $\ell'\in\OPT$ of some
level $j'\ge j$ stabs a \kshape{} in
$\K_{\ell}$ only if the distance between $\ell$ and $\ell'$ is at most $\alpha\eps^{j-1}$.
Another key other insight is that since the ratio between the widest and the narrowest part of any
$K\in\K_\ell$ is $1/\delta$, the solution we compute is also a $O(1/\delta+\eps)$-approximate
solution.  Using the above facts, we claim that if we choose our offset $r$ uniformly at random from
the range $\{0,\eps^d, 2\eps^2d,\dots,\alpha\eps^{-2}\}$, then the overall cost of these \textit{line
operations} is only $O(\eps)\cdot\OPT$.
Further to bound the cost of the add operations, we note that each add operation is either done on a
segment in $\OPT$, or is an operation that created a subproblem. We will show that we can charge the
latter operations to segments from $\OPT$ inside the subproblem thus created, whose total cost is at
least $\eps^{-1}$ times the width of the subproblem.
We refer to Appendix~\ref{app:ptas} for a formal description of our analysis.




\begin{restatable}{lemma}{errorbound} \label{lem:errorbound}
    There is a discrete value for the offset $r \in \{0, \eps^d, 2\eps^d,\ldots,\eps^{-2}\}$
    such that the described sequence of operations produces a solution of cost
    at most $(1+O(\eps))\OPT$.
\end{restatable}

\begin{restatable}{theorem}{thmptas} \label{thm:ptas}
    For each constant $k\in\N$ there is a PTAS for the \kstabbing{} problem when each given
    \kshape{} consists of pieces of a constant range of widths that are placed strictly on top of
    each other.
\end{restatable}

\section{General case}\label{sec:general-case}

In this section, we study the general case of stabbing rectilinear
polygons. In contrast to the cases studied in Sections~\ref{sec:qptas} and \ref{sec:ptas},
we show that it does \emph{not} admit a $(1+\eps)$-approximation
algorithm, even for only slightly more general types of instances.

\subsection{APX-hardness}

Formally, we prove that stabbing is $\mathsf{APX}$-hard, already if each input
polygon is a 3-shape.

\begin{theorem} \label{thm:stabbing-APX-hard}The stabbing problem
for 3-shapes is $\mathsf{APX}$-hard.

\end{theorem}

On the other hand, any 2-shape satisfies the hourglass property; hence,
stabbing is unlikely to be $\mathsf{APX}$-hard for this class of objects since
we have a QPTAS for this case.

\begin{restatable}{proposition}{trivhourglass} \label{prop:trivhourglass}
    Each 2-shape satisfies the hourglass property.
\end{restatable}

In the remainder of this subsection, we prove Theorem~\ref{thm:stabbing-APX-hard}.
We give an
L-reduction \adi{(with $\alpha=\beta=1$)} from the vertex cover problem to stabbing for
3-shapes. Note that it is $\mathsf{NP}$-hard to approximate vertex cover with
a strictly better approximation factor than $\sqrt{2}$ \cite{KhotMS18}.
We will obtain the same lower bound for stabbing.

Consider a given instance $G=(V,E)$ of vertex cover. Remember that in vertex
cover, we are required to select a subset $S\subseteq V$ of smallest
size such that for each $e\in E$ one of its end points is in $S$.
We construct an instance of \kstabbing{} corresponding to $G$ as
follows. Assume that $V=\{v_{1}, \dots,v_{n}\}$. For each $v_{i}\in V$
construct a $1\times1$ square $s_{i}$, such that they are all arranged
in a column separated by $1$ unit distance each (see Figure~\ref{fig:vcreduction}).
Formally, for each $v_{i}\in V$ the top-left corner of the square
$s_{i}$ has the coordinates $(0,2i-1)$. Note that the squares $s_{1}, \dots,s_{n}$
do \emph{not }belong to our input shapes, but they only help us to
construct the latter. For each edge $\{v_{i},v_{j}\}\in E$ we define
a $3$-shape $r_{i,j}$ as the union of the three rectangles $s_{i},[0,n+1]\times[2i-1,2j-2]$
and $s_{j}$ (see Figure~\ref{fig:vcreduction}).

\begin{figure}
\centering
\includegraphics[page=1]{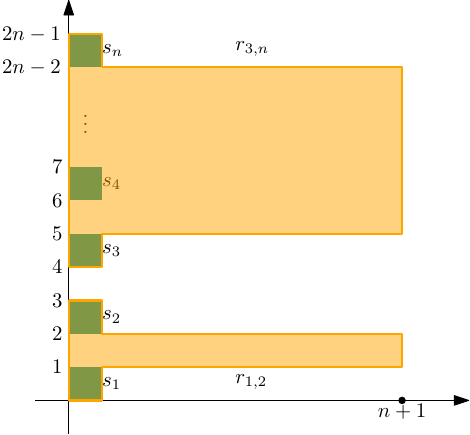}
\caption{Construction of $k$-\stabbing{} instance in our reduction from vertex cover.}
\label{fig:vcreduction} 
\end{figure}

Note that none of the resulting shapes satisfies the hourglass property,
and also for neither of them the widths of its three rectangles are
in a constant range. The width of the widest rectangle of each constructed
3-shape is greater than $n$, but there is always a feasible solution
with cost $n$ that simply stabs the square $s_{i}$ for each vertex
$v_{i}\in V$. Thus, in any given solution to the stabbing instance,
we can assume w.l.o.g.~that no 3-shape is stabbed across its widest
rectangle.

\begin{restatable}{lemma}{vcreduction} \label{lem:vcreduction} For each $\gamma\in\mathbb{N}$,
the given instance of vertex cover instance has a solution of size
$\gamma$ if and only if the corresponding \kstabbing{} instance has
a solution of cost $\gamma$.
\end{restatable}

This yields the proof of Theorem~\ref{thm:stabbing-APX-hard}.

\subsection{Set-Cover hardness}

In this section, we further show that \kstabbing{} for arbitrary
\kshapes{} cannot be approximated with a ratio of $o(\log n)$, unless
$\mathsf{P}=\mathsf{NP}$. In fact, we show that the problem is as hard as general instances
of \textsc{Set Cover}, for which it is known that it does not admit an $o(\log n)$-approximation
algorithm, unless $\mathsf{P}\ne\mathsf{NP}$~\cite{dinur2014analytical}.


\begin{theorem} \label{thm:stabbing-log-n-hard}The \stabbing{} problem
for \kshapes{} does not admit an $o(\log n)$-approximation algorithm.

\end{theorem}

In the remainder of this subsection, we prove Theorem~\ref{thm:stabbing-log-n-hard}.
We reduce from the hitting set problem, which is known to be equivalent
to \textsc{Set Cover}. In hitting set, we are given a
set of elements $P=\{p_{1},p_{2},\ldots,p_{n}\}$, and a family of
their subsets $\calF$. The aim is to compute a minimum size subset
$H\subseteq P$ such that every set in $\calF$, contains a point
in $H$.

Our construction here is similar to the reduction from vertex cover
above. For each element $p_{i}\in P$, we construct a unit square
$s_{i}$ with its top left corner being located at $(0,2i-1)$. For
each set $S_{i}=\{p_{i_{1}},p_{i_{2}},\ldots,p_{i_{f}}\}\in\calF$
we construct a \kshape{} which is a stack of the rectangles $s_{i_{1}}$,
$[0,n+1]\times[2i_{1}-1,2i_{2}-2]$, $s_{i_{2}}$, $[0,n+1]\times[2i_{3}-1,2i_{3}-2]$,$s_{i_{3}}$,
\dots, $s_{i_{f-1}}$,
$[0,n+1]\times[2i_{f-1}-1,2i_{f}-2]$, $s_{i_{f}}$ (see Figure~\ref{fig:screduction}).

\begin{figure}
    \centering
    \includegraphics[page=2]{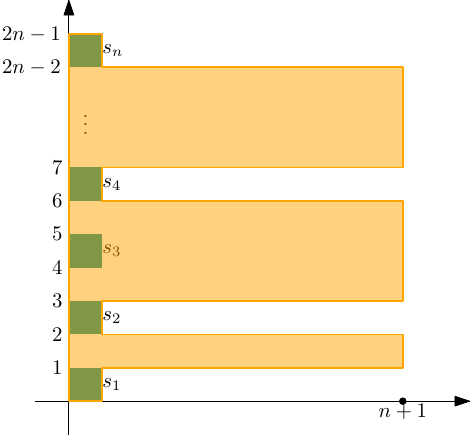}
    \caption{The \kshape{} constructed for set $S=\{v_{1},v_{2},v_{4},v_{n}\}\in\calF$}
    \label{fig:screduction}
\end{figure}

Again, these constructed \kshapes{} neither satisfy the hourglass
property nor the condition that the widths of its rectangles are in
a bounded range. As before, for each given solution, we can assume
w.l.o.g. that no \kshape{} is stabbed across one of its wide rectangles
(i.e., of width $n+1$).

\begin{restatable}{lemma}{screduction} \label{lem:screduction} For each $\gamma\in\mathbb{N}$,
the given instance of hitting set has a solution of size $\gamma$ if and
only if the constructed \kstabbing{} instance has a solution of
cost $\gamma$. \end{restatable}

Therefore, we constructed an appropriate approximation preserving solution from
\textsc{Set Cover}, which yields the proof of Theorem~\ref{thm:stabbing-log-n-hard}.

\subsection{Approximation algorithm}

We show that there is a polynomial time $O(k)$-approximation algorithm
for stabbing for \kshapes{}. Our algorithm is \emph{LP-relative,
}meaning that it outputs a solution whose cost is at most by a factor
of $O(k)$ larger than the cost of the optimal solution to the canonical
LP-formulation for the problem.

Suppose we are given an instance $\calK$ of \kstabbing{} with $n:=|\calK|$.
In principle, there is an infinite set of possible line segments that
we could use for our solution. However, it is sufficient to restrict ourselves
to a polynomial number of line segments which we construct using the
following lemma.

\begin{restatable}{lemma}{candidates} \label{lemma:candidates}
In polynomial time, we can
construct a set $\calC$ of line segments with the following properties:
\begin{itemize}
    \item $\calC$ contains $O((kn)^{3})$ segments,
\item $\calC$ contains no redundant segments, where a segment is \emph{redundant}
if it stabs exactly the same \kshapes{} as another segment, or no
\kshapes{} at all, and
\item $\calK$ admits an optimal solution using only the segments from $\calC$.
\end{itemize}
\end{restatable}

%

Using $\calC$, we define a linear program that corresponds to $\calK$.
\begin{align}
\min & \sum_{s\in\calC}|s|\cdot z_{s}\nonumber \\
\text{s.t.} & \sum_{s\in\calC\colon s\,\mathrm{stabs\,\,}K}z_{s}\geq1 &  & \forall K\in\K\label{eq:LP-stabbing}\\
 & z_{s}\geq0 &  & \forall s\in\calC\eqend\nonumber
\end{align}
If each \kshape{} $K\in\K$ is a rectangle, then it was shown by
Chan et al.~\cite{ChanD0SW18} that this LP has a constant integrality
gap.
\begin{theorem}[\cite{ChanD0SW18}]
\label{thm:Chan}If each \kshape{} $K\in\K$ is a rectangle, then
there is a constant $\alpha$ such that for any solution $z$ to LP~\eqref{eq:LP-stabbing},
in polynomial time we can compute an integral solution to~\eqref{eq:LP-stabbing}
whose cost is at most $\alpha\sum_{s\in\calF} |s|\dot z_{s}$.
\end{theorem}
Using Theorem~\ref{thm:Chan}, we construct now an $(\alpha\cdot k)$-approximation algorithm
for arbitrary \kshapes{}. Suppose we are given an optimal solution
$z^{*}$ to the LP~\eqref{eq:LP-stabbing}. We define a new solution
$\tilde{z}$ by setting $\tilde{z}_{s}:=k\cdot z_{s}^{*}$ for each
segment $s\in\calF$. Each \kshape{} $K\in\K$ is composed out of
at most $k$ rectangles. Thus, for each \kshape{} $K\in\K$ there
is one of these rectangles $R$ for which $\sum_{s\in\calF\colon s\,\mathrm{stabs\,\,}R}z_{s}^{*}\geq1/k$
and, therefore, $\sum_{s\in\calF\colon s\,\mathrm{stabs\,\,}R}\tilde{z}_{s}\geq1$.
Let $\mathcal{R}$ denote the set of all these rectangles for all
\kshapes{} in $\K$. We apply
Theorem~\ref{thm:Chan} on $\tilde{z}_{s}$ and $\mathcal{R}$ which
yields a set of segments $\tilde{\mathcal{S}}$ whose cost is at most
$\alpha\cdot\sum_{s\in\calF}|s|\cdot\tilde{z}_{s}=\alpha k\cdot\sum_{s\in\calF}|s|\cdot z_{s}^{*}\le\alpha
k\cdot\OPT$. Since $\tilde{\mathcal{S}}$ stabs $\mathcal{R}$, it also stabs $\K$.
Hence, $\tilde{\mathcal{S}}$ yields an $O(k)$-approximation to our problem.

\begin{theorem}There is a polynomial time $O(k)$-approximation algorithm for \kstabbing{}.
\end{theorem}
We remark that our algorithm extends also to the setting in which each given shape consists of
at most $k$ rectangles that are not necessarily connected, but such that still at least one
of them needs to be stabbed.

\bibliography{references}

\appendix
\section{Missing proofs from Section~\ref{sec:qptas}} \label{app:qptas}

\subsection{Reduction to Set Cover} \label{app:screduction}
In this subsection we show that the \kstabbing{} problem can be
reduced to a general weighted set cover instance.  We are given an instance of \kstabbing{} with a
set of $n$ \kshapes{}, $\K$. Since each \kshape{} can be described by at most $2k$ points (2 each
describing each of its constituent rectangular sections), we have a total of $2nk$ points, and only
$\binom{2nk}{2}=O(n^2k^2)$ possible combinatorially distinct segments in the solution.

We create an instance of set cover as follows. The universe of elements is given by the set of all
\kshapes. Now for each of the $O(n^2k^2)$ possible segments $\ell$ in the solution, we create a set
of weight $|\ell|$, which contains all the \kshapes{} that the segment stabs. This forms our family of
subsets. Now a solution of weight $w$ to the set cover instance, corresponds to a set of segments of
weight $w$ which stab all \kshapes{} in $ \K$. Hence, we have an approximation preserving reduction
from \kstabbing{} to weighted set cover, and by extension, a $(\log n)$-approximate algorithm for
\kstabbing{}.

\subsection{Proof of Lemma~\ref{lem:preprocessing}}
\qptaspreprocess*
\begin{proof}
    Using the \((\log n)\)-approximation algorithm, obtain a solution to \(\calK\) and determine its
    cost \(\gamma \in [\OPT(\calK), \log n\cdot \OPT(\calK)]\). Consequently, \(\OPT(\calK) \in
    [\frac{1}{\log n} \gamma, \gamma]\).

    Scaling each \kshape{} in \(\calK\) along the \(x\)-axis by a factor of \(\beta \coloneqq \log n
    \frac{1}{\gamma}\) yields an instance \(\Kp\) with \(\OPT(\Kp) = \beta \cdot \OPT(\calK) \in [1,
    \log n]\), which has the following implications for solving \(\Kp\):
    \begin{itemize}
        \item Any parts \(R\) of \kshapes{} in \(\Kp\) with \(w(R) > \log n\) can be discarded,
            since an optimal solution cannot stab them.
        \item Greedily stabbing all \kshapes{} \(K \in \Kp\) with \(\wmin(K) \leq
            \smash{\frac{\eps}{n}}\) requires segments of total length at most \(\eps \leq \eps
            \cdot \OPT(\Kp)\).
    \end{itemize}

    An algorithm for solving \(\Kp\) can perform this preprocessing and continue to operate only on
    \kshapes{} \(K\) with \(\smash{\frac{\eps}{n}} < \wmin(K) \leq \wmax(K) \leq \log n\).
    Furthermore, it can be assumed w.l.o.g. that \(\wtotal(\Kp) \leq n\log n\); otherwise, \(\Kp\)
    could be partitioned into independent sub-instances to be solved separately.
\end{proof}

\subsection{Proof of Lemma~\ref{lem:partition1}}
\strippartition*
\begin{proof}
    For the sake of conciseness, define \(w\coloneqq\wmax(\K)\).
    Define \(Z \coloneqq \{i \cdot \mu / n \in [0, w / \mu) \mid i \in \Z\}\) to be the set of offsets.
    For every \(z \in Z\), \(\Lz \coloneqq \{\{z + i \cdot w / \mu\} \times \R \mid i \in \Z\}\)
    is a set of uniformly spaced vertical grid lines, and let \(\calK_{\text{rest},z} \subseteq \calK\) be the set of
    \kshapes{} that are intersected by a line from \(\Lz\). For simplicity we shall henceforth
    use $\Ez$ where $z$ is clear from the context.
    Also let the set of \kshapes{} completely
    contained in $[z + i \cdot w / \mu,z + (i+1) \cdot w / \mu]$ belong to the set $\K_{i+1}$, and
    similarly let $K_0$ be the set of all \kshapes{} fully contained inside the strip $[0,z]$.
    Clearly each of the $\K_i$ satisfies that $\wtotal(\calK_{i})\leq\wmax(\K)/\mu$
    (proving property $(iii)$ ).

    Clearly the $\K_i$ sets are disjoint subsets of $\K$, so any solution (and in particular the
    optimal solution) to $\K$ already stabs $\K_i$ for all $i$. This gives us that
    $\OPT\ge\sum_{\ell=0}^{t}\OPT(\K_{\ell})$ (proving property $(i)$).

    We now need to show that there is a choice of \(z \in Z\) such that \(\OPT(\Ez) \leq 8 \mu \cdot
    \OPT(\calK)\). To this end, suppose \(\SOPT\) is an optimal solution to the entire instance
    \(\calK\), and fix some choice of \(z \in Z\). The idea is to collect all (parts of) segments
    from \(\SOPT\) that are needed in order to stab \(\Ez\), and estimate their total cost.

    Since no \kshape{} is wider than \(w\), every \kshape{} in \(\Ez\) must be entirely contained
    within the \([-w, +w]\)-strip around some vertical line from \(\Lz\). It is therefore sufficient
    to collect all intersections of segments in \(\SOPT\) with such strips, obtaining a set \(\SEz\)
    that fully stabs \(\Ez\) (This differs from the case of stabbing rectangles, where only the
    segments directly intersected by \(\Lz\) are needed).

    For estimating the cost of \(\SEz\), denote by \(\Czs\) the total cost of the segments in
    \(\SEz\) generated by \(s \in \SOPT\). Notice that \(s = [x_1, x_2] \times \{y\}\) intersects
    the \([-w, +w]\)-strip around a line \(\ell = \{x_\ell\} \times \R\), if and only if \(x_\ell
    \in [x_1 - w, x_2 + w]\). Counting the intersections of \(s\) with those strips is therefore
    equivalent to counting the intersections of a segment of length \(|s| + 2w\) with the lines
    themselves.

    There are two mutually exclusive cases:

    \begin{description}
        \item[Case 1.] \(|s| + 2w \geq w / \mu\).\nopagebreak

        This implies that there is at least one intersection. Because the distance between the lines
        in \(\Lz\) is \(w / \mu\), there are at most \(\frac{|s| + 2w}{w / \mu} + 1\) intersections,
        each of which costs at most \(\min\{|s|, 2w\}\). Therefore,
        \begin{align*}
            \Czs &\leq \left(\frac{|s| + 2w}{w / \mu} + 1\right) \cdot \min\{|s|, 2w\}\\
            &= \left(\frac{|s| + 2w}{w / \mu} + \frac{w / \mu}{w / \mu}\right) \cdot \min\{|s|, 2w\}\\
            &\leq \left(\frac{|s| + 2w}{w / \mu} + \frac{|s| + 2w}{w / \mu}\right) \cdot \min\{|s|, 2w\}\\
            &= (2 / w) \mu |s| \cdot \min\{|s|, 2w\} + 4 \mu \cdot \min\{|s|, 2w\}\\
            &\leq 4 \mu |s| + 4 \mu |s|\\
            &= 8 \mu |s|\eqend
        \end{align*}

        \item[Case 2.] \(|s| + 2w < w / \mu\).\nopagebreak

        This implies that there is at most one intersection. To determine its probability, consider
        the set \(\mathcal{L} \coloneqq \bigcup_{z \in Z} \Lz = \{k \cdot \mu / n \mid k \in \Z\}\).
        Since this is a disjoint union, each line \(\ell \in \mathcal{L}\) can be associated with a
        unique value of \(z\) it was generated by, i.e., there exists a unique \(z \in Z\) such that
        \(\ell \in \Lz\).

        Now counting how many choices of \(z\) produce an intersection is equivalent to counting the
        intersections of \(\mathcal{L}\) with the (elongated) segment: Because \(|s| + 2w < w /
        \mu\), it cannot intersect more than one line from the same \(\Lz\), precluding double
        counting.

        Hence there are at most \(\frac{|s| + 2w}{\mu / n} + 1\) choices of \(z\) producing an
        intersection. Division by the total number of choices \(|Z| = wn / \mu^2\) yields
        \begin{align*}
            \Exp_z[\Czs] &\leq \left(\frac{|s| + 2w}{\mu / n} + 1\right) \cdot \frac{1}{wn / \mu^2} \cdot \min\{|s|, 2w\}\\
            &= \left(\frac{\mu}{w} \cdot |s| + 2\mu + \frac{\mu^2}{wn}\right) \cdot \min\{|s|, 2w\}\\
            &= \frac{\mu}{w} \cdot |s| \cdot \min\{|s|, 2w\} + 2\mu \cdot \min\{|s|, 2w\} + \frac{\mu^2}{wn} \cdot \min\{|s|, 2w\}\\
            &\leq 2 \mu |s| + 2 \mu |s| + \frac{\mu^2}{wn} \cdot 2w\\
            &= 2 \mu |s| + 2 \mu |s| + 2 \mu \cdot \frac{\mu}{n}\\
            &\leq 2 \mu |s| + 2 \mu |s| + 2 \mu |s| \tag{since~\(\frac{\mu}{n} \leq \wmin(\calK) \leq |s|\)}\\
            &= 6 \mu |s|.
        \end{align*}
    \end{description}
    The above argument shows that $\Exp[\OPT(\K_{\mathrm{rest}})]\le O(\mu)\cdot\OPT$, and hence
    there is at least one offset satisfying $\OPT(\K_{\mathrm{rest}})\le O(\mu)\cdot\OPT$. For $\mu$
    that is polynomial in $n$, there are only a polynomial $\frac{w/\mu}{\mu/n}=\frac{wn}{\mu^2}$
    number of possible offsets, and hence we can guess it in polynomial time.
\end{proof}

\subsection{Proof of Lemma~\ref{lem:partition2}}
\blockpartition*
\begin{proof}
Since we use only horizontal segments to stab \kshapes{} w.l.o.g. (by scaling along the $y$
direction) we can assume that the at most $2kn$ points describing the instance occupy
consecutive integral $y$-coordinates, starting at $y=0$.

Consider the segments from $\OPT(\K_i)$. Starting from $y=0$ and going up, we can start counting
the cumulative cost of segments in $\OPT$. Let $h$ be the $y$-coordinate at which this cumulative cost crosses
$\OPT(\K_{i})/2$, and $s=[a,b]\times \{h\}$ be the corresponding segment. Since the width
of $S_{i}$ is at most $\wmax(\K)/\mu$, no segment in $\OPT(\K_{i})$ is
wider than $\wmax(\K)/\mu$. From this we can infer that cost of segments from $\OPT(\K_i)$, below (and similarly, above)
the segment $s$ should have been at least $\OPT(\K_{i})/2-\wmax(\K)/\mu$.

Since there are only a polynomial $2kn$ number of possible $y$-coordinates, we can guess
this value $h$ in polynomial time by enumeration.
\end{proof}

\subsection{Proof of Lemma~\ref{lem:horizontal}}
\horizontal*
\begin{proof}
    After our sequence of (correctly guessed) balanced horizontal cuts, let us assume that there
    are $t$ connected components, with cost at least $\wmax(\K_i)/2\mu^2-\wmax(\K)/\mu$. This can
    happen only if there were $t-1$ such cuts applied. If we charge the cost of every cut $s$ to
    the cost of segments of $\OPT(\K_i)$ within a cell $C$, we get
    \[ \frac{|s|}{\OPT(C)}=\frac{\wmax(\K)/\mu}{\wmax(\K)/2\mu^2-\wmax(\K)/\mu} = \frac{2\mu}{1-2\mu} \le 3\mu.\]
    Where the last inequality follows under the assumption of $\mu\le\eps<1/3$.

    Summing over all such horizontal cuts, we get the total cost to be
    at most $3\mu\cdot\wmax(\K_i)$.

\end{proof}
\subsection{Proof of Lemma~\ref{lem:corrguess}}
\corrguess*
\begin{proof}
    The bound on the cost of segments directly follows from Lemma~\ref{lem:partition1} and
    Lemma~\ref{lem:horizontal}. The above Lemmas also guarantee that the number of guesses for
    applying them is polynomial. The overall number of guesses is $n^{O(\eps^{-3}\log^2 n)}$ because while
    guessing the segments within a cell (formed as a result of Lemmas~\ref{lem:partition1} and
    Lemmas~\ref{lem:horizontal}) there are only a polynomial number of (combinatorially distinct)
    possible segments, and there are at most $\eps^{-3}\log^2n$ such segments within each cell.
\end{proof}

\section{Missing proofs and details from Section~\ref{sec:ptas}} \label{app:ptas}


\subsection{Constant Factor Approximation} \label{app:deltaapprox}
In this subsection, we give a simple $O(1/\delta)$-approximation algorithm for \kstabbing{}
when for each \kshape{} the widths of any two of its rectangles differ by at most a factor of $1/\delta$.
Given such an instance of \kstabbing{} and let $\OPT$ denote its optimal solution. We create an
instance of \stabbing{} of rectangles as follows. For each given \kshape{} $K$, we take the rectangle
of smallest width and height that contains $K$. We add all these rectangles to our constructed instance of \stabbing{}
for rectangles. Let $\OPT'$ denote the optimal solution to that instance.
On this instance, we apply the known PTAS for \stabbing{} for rectangles
\cite{KhanSW22}.

We observe that $\OPT' \le O(\OPT /\delta)$ since we can simply take each segment in $\OPT$ and extend it by
a factor of $1/\delta$ in each direction. Due to our assumption about the widths of the input rectangles, this yields a feasible
solution to our initially given instance to  \kstabbing{}. Therefore, the solution computed by our PTAS yields a solution
with cost at most $O(\OPT /\delta)$, and hence a $O(1/\delta)$-approximation.

%
%

\subsection{Proof of Lemma~\ref{lem:discretize}, Preprocessing Step}
We note here that the discretization steps here are similar to the case of rectangles as done
by Khan, Subramanian, and Wiese \cite{KhanSW22}.
\discretize*
\begin{proof}
    Using the \(\alpha\)-approximation algorithm, obtain a solution to \(\calK\) and determine its
    cost \(C \in [\OPT(\calK), \alpha \OPT(\calK)]\). Consequently, \(\OPT(\calK) \in
    [\frac{1}{\alpha} C, C]\).

    Scaling each \kshape{} in \(\calK\) along the \(x\)-axis by a factor of \(\beta \coloneqq (1-2\eps)
    \frac{\alpha}{C}\) yields an instance \(\Kp\) with \(\OPT(\Kp) = \beta \cdot \OPT(\calK) \in [(1-2\eps),
    (1-2\eps)\alpha]\). Now, greedily stabbing all \kshapes{} \(K \in \Kp\) with \(\wmin(K) \leq
    \frac{\alpha\eps}{n}\) requires segments of total length at most \(\alpha\eps\). This accounts
    for a factor
    \[ \frac{\SOL}{\OPT}\le\frac{\OPT+\alpha\eps}{\OPT}\le1+\frac{\alpha\eps}{1-2\eps}
    \le1+\eps\cdot\frac{\alpha}{1-2(1/3)}=1+O(\eps),\]
    increase in the cost of the solution.

    Now we try to discretize the $x$-coordinates of all \kshapes. Each \kshape{} consists of up to
    $k$ rectangular parts. Extend each such part on both sides to make their $x$-coordinates align
    with the next nearest multiple of ${\eps}^{d}$. Since this involves extension by at most
    $2\eps^d<2\eps/n$ to the width of every rectangular section, the total cost of the solution
    increases by at most a factor of $1+2\eps$.

    In other words we can also say that the cost of an optimal solution goes up, to \(\OPT
    (\Kp) \in [1, \alpha]\). Now we notice that any parts \(R\) of \kshapes{} in \(\Kp\) with \(w
    (R) > \alpha\) can be discarded, since an optimal solution cannot stab them. Hence, the above
    steps ensure that $\frac{\alpha\eps}{n}<\wmin(\calK)\le\wmax(\calK)\le \alpha$
    (proving property $(i)$), and that each $x$-coordinate of any \kshape{} is discretized.

    Since we stab the \kshapes{} using horizontal lines, we can stretch the a \kshape{} vertically
    without affecting the solution cost. Since each \kshape{} has at most $k+1$ distinct
    $y$-coordinates, there are at most $(k+1)n$ distinct $y$-coordinates in the problem instance,
    and we can stretch the instance in such a manner that these $y$-coordinates are consecutive
    integers between $0$ to $(k+1)n$, ensuring that all $y$-coordinates are also discretized
    (proving property $(ii)$).

    It can be assumed w.l.o.g. that \(\wtotal(\Kp) \leq \alpha n\); otherwise, \(\Kp\)
    could be partitioned into independent sub-instances to be solved separately. Similarly since
    the $y$-coordinates are also already shown to lie between $0$ and $2kn$, we have also shown
    property $(iii)$.

    We now consider the final property. Let $\OPT'$ be the optimal solution to the instance obtained
    after applying the three steps above. Then the cost of $\OPT'$ can exceed the cost of the
    original optimal solution $\OPT$ by at most a factor of $1+O(\eps)$. Consider any horizontal
    segment $\ell\in\OPT'$ that is longer than $\alpha/\eps$. From the left endpoint, we divide the
    segment into consecutive smaller segments of length $\alpha/\eps-2\alpha$ each, with one potential last
    piece being smaller than $\alpha/\eps-2\alpha$. Now, for each smaller segment, we extend it on both sides
    in such a way that it completely stabs the \kshapes that it intersects (i.e., that it
    intersects before the extension).  Since the maximum width of a \kshape is $\alpha$, we extend each
    such segment by at most $2\alpha$ units. We denote by $|\ell|$ the length of $\ell$ and conclude that
    we increase the length of $\ell$ by at most a factor of
    \begin{align*}
        \left(|\ell|+\left\lceil\frac{|\ell|}{\alpha/\eps-2\alpha}\right\rceil\cdot2\alpha\right) \cdot \frac{1}{|\ell|}
        &\le 1+\frac{2\alpha}{|\ell|}\cdot\left\lceil\frac{\eps|\ell|}{\alpha-2\alpha\eps}\right\rceil \\
        &\le 1+\frac{2\alpha}{|\ell|}\cdot\left(1+\frac{\eps|\ell|}{\alpha-2\alpha\eps}\right) \\
        &\le 1+2\eps+2\eps\cdot\frac{\alpha}{\alpha-2\alpha\eps} \tag{since $|\ell|>\alpha/\eps$}\\
        &\le 1+2\eps+2\eps\cdot\frac{1}{1-2/3} \tag{since $\eps<1/3$}\\
        &\le 1+8\eps.
    \end{align*}

This completes the proof of the lemma.
\end{proof}

\subsection{Proof of Lemma~\ref{lem:ptasruntime}}
\ptasruntime*
\begin{proof}
    We know that the DP always picks the sequence of operations that gives a result of minimum
    cost. Since there is a sequence of operations that produces a solution of cost $(1+\eps)\OPT$,
    the DP returns a solution of cost at most that.

    Now let us consider the running time of the algorithm. Since a DP problem is defined on a
    discrete rectangular cell, there  are at most

    \[\frac{\alpha n}{\eps^d} \times (k+1)n \le
    \alpha(k+1)\eps^{-3}n^3
    \tag{since $\eps^3/n<\eps^d\le\eps^2/n$}\]

    possibilities for a corner vertex of a rectangle, and hence
    $\binom{\alpha(k+1)\eps^{-3}n^3}{2}=O(k^2n^{6}/\eps^{6})$ possible rectangles.

    Similarly, the subproblem definition also includes a set of segments $\cL$ of size at most
    $3\eps^{-3}$. Since the segments are discrete we can count the number of horizontal segments
        by picking two points (corresponding to the starting and ending points) from the available
        discrete points on a horizontal line, and then fixing its $y$-coordinate.
    \[\binom{\alpha n/\eps^d}{2}\times(k+1)n
    \le\binom{\alpha n^2/\eps^3}{2}\times(k+1)n
    \le \frac{\alpha^2(k+1)n^5}{\eps^6}
    \]
    that is, $O (kn^5/\eps^6)$ possible segments. So there are at most $(kn/\eps)^ {O(1/\eps^3)}$ subsets of
segments $\cL$ of size at most $3\eps^{-3}$, and at most $(kn/\eps)^{O(1/\eps^3)}$ valid DP-cells.

For each DP-cell, we have to consider all possible candidate solutions and select the
minimum. There are at most $\alpha n/\eps^d = O(n^2/\eps^3)$ possible \textit{line
operations} and $(kn/\eps)^{O(1/\eps^3)}$ possible add
operations.
Note that if the DP performs a trivial operation, then there is no choice to make here, but the
trivial operation is selected automatically.

Hence, the total number of possible operations for a given DP-cell is $(kn/\eps)^{O(1/\eps^3)}$. For
each \textit{line operation} we call the $O(1)$-approximation algorithm
which is also runs in polynomial time \cite{ChanD0SW18}.  Since we have $(kn/\eps)^{O(1/\eps^3)}$ DP-cells, our overall running time is bounded by~$(kn/\eps)^{O (1/\eps^3)}$.
\end{proof}

\subsection{Proof of Lemma~\ref{lem:algn}}
    \alignedopt*
\begin{proof}
    A segment $\ell$ in some level $j$ will be of length in $(\alpha\eps^j, \alpha\eps^{j-1}]$.
    To align it to a grid line of level $j+3$ we would need to extend it by at most
    $\alpha\eps^{j+1}$ on each side. The new segment $\ell'$ thus obtained is of length
    \[ |\ell'| \le |\ell| + 2\alpha\eps^{j+1} \le |\ell| \left(1+\frac{2\alpha\eps^{j+1}}{|\ell|}\right)
    <|\ell|\left(1+\frac{2\alpha\eps^{j+1}}{\alpha\eps^j}\right)=|\ell|\cdot(1+2\eps).\]
    Therefore, the sum of weights over all the segments in $\OPT$ is \[\sum_{\ell\in\OPT}
    |\ell'|\le \sum_{\ell\in\OPT} |\ell|\cdot(1+2\eps) = (1+2\eps)\cdot\OPT.
    \qedhere
    \]
\end{proof}

\subsection{Approximation Factor: Proof of Lemma~\ref{lem:errorbound}}
\errorbound*
\begin{proof}
    In the described sequence of operations, some add operations are applied on segments from $\OPT$
    (or their parts), and
    hence the cost across all such add operations is at most cost of $\OPT$. Similarly, all trivial
    operations are applied on segments that were `added' before, and hence their cost is also
    already accounted for. So we are left with analyzing the cost of stabbing the rectangles which
    are intersected by the lines along which we apply the line operations, and the add operations
    which are not which are not applied on a segment from $\OPT$ (that is, the ones that are used
    to partition a vertical strip). We claim that for a
    discretized random offset $r \in \{0, \eps^d, 2\eps^d,\ldots,\alpha\eps^{-2}\}$, the expected cost is $O(\eps\cdot\OPT)$, which
    would give us the required result.


    Let us first consider any add operation of level $j$ that is applied to a horizontal line
    $\ell$ (that is not in $\OPT$). We do such an operation only after accounting for $\eps^
    {-3}$ segments from $\OPT$ of level $j$, that is, segments of cost at least $\eps^
    {-3}\cdot\alpha\eps^j = \alpha\eps ^{j-3}$. Since a segment of width $\alpha\eps^{j-2}$
    (width of strip) is sufficient to stab all the \kshapes{} stabbed by $\ell$, we see that this
    horizontal segment only takes $\eps$ times the cost of the segments in $\OPT$ that we have
    already accounted for. We charge the cost of adding this segment to the solution, to the $\eps^
    {-3}$ segments from $\OPT$ that were counted before adding it. Since any segment in $\OPT$
    gets charged only once, we can infer that the
    cost of such add operations is at most $2\eps\cdot\OPT$.

    Now, let us consider the \textit{line operations} applied to vertical grid lines. Consider a grid line $\ell$ of level $j$. We wish to bound the
    cost of stabbing all the rectangles in $\cR_\ell$ intersected by grid lines, over all
    levels $j$. Let $\OPT_{\cR_\ell}$ be the set of minimum cost that stabs all segments in $\mathcal
    {R}_\ell$, and let $\OPT_{K_\ell}$ be the set of segments from $\OPT$ that stab the
    corresponding \kshapes{} in $\K_\ell$. From Appendix~\ref{app:deltaapprox}, we know that $\OPT_{\cR_\ell} \le
    1/\delta\cdot\OPT_{K_\ell}$. So instead of bounding the cost of stabbing $\cR_\ell$, we
    instead focus on bounding the cost of stabbing $\K_\ell$.

    A horizontal line segment $\ell'\in\OPT$ of some level
    $j'\ge j$ stabs a \kshape{} in $\K_{\ell}$ only if the distance between $\ell$ and $\ell'$ is
    at most $\alpha\eps^{j-1}$. So we need to bound the cost, over all levels $j$, of line
    segments of level $j$ in $\OPT$ (call this set $\OPT_j$) intersected or close to grid lines of
    level $j$. For a horizontal segment $\ell'\in\OPT_j$, let $I_{\ell'}$ be the indicator variable
    representing the event that $\ell'$ is within distance $\alpha\eps^{j-1}$ of
    $\ell$. Now, since any two consecutive grid lines of level $j$ are separated by
    $\alpha\eps^{j-2}$, there are $(\alpha\eps^{j-2})/\eps^d$ possible shifts for these
    grid lines due to our
    offset, and each of these shifts has the same probability. Similarly, $(2\alpha\eps^{j-1}/\eps^d)+1$
    of these offsets would allow $\ell'$ to be within $\alpha\eps^{j-1}$ distance of $\ell$.
    So, if we take a random discrete offset $r  \in \{0, \eps^d, 2\eps^d,\ldots,\eps^{-2}\}$, we have that
        \[\Exp[I_{\ell'}] \le \frac{(2\alpha\eps^{j-1}/\eps^d)+1}{\alpha\eps^{j-2}/\eps^d}
        = 2\eps+\frac{\eps^{2+(d-j)}}{\alpha}\le3\eps. \tag{since $j\in\{0,1,\dots,d-1\}$}\]

    With the expectation computed above, we can upper bound the expected cost of segments in
    $\OPT_{\K_\ell}$ as:
    \begin{align*}
       \Exp\left[\sum_j\sum_{\ell\in\OPT_j} I_\ell\cdot|\ell|\right]
        &= \sum_j\sum_{\ell\in\OPT_j}
        \Exp\left[I_\ell\cdot|\ell|\right] \\
        &= \sum_j\sum_{\ell\in\OPT_j} |\ell|\cdot\Exp[I_\ell] \\
        &\le \sum_j\sum_{\ell\in\OPT_j} |\ell|\cdot (3\eps) \\
        &=  3\eps\cdot\OPT
    \end{align*}
    Since $\OPT_{\cR_\ell} \le 1/\delta\cdot\OPT_{K_\ell}$, we get that
    $\Exp[\OPT_{\cR_\ell}] \le 3\eps/\delta\cdot\OPT$. And since we use a PTAS (which, let's
    say has an approximation factor of $(1+\eps')$) for computing
    the cost of $\OPT_{\cR_\ell}$, the solution returned by our algorithm takes an additional
    cost of $\frac{3\eps(1+\eps')}{\delta}\cdot\OPT$.
\end{proof}


Now we prove our main theorem.

\begin{proof}
We gave a DP algorithm in Section~\ref{sec:ptas} which was shown to have the required running time
in Lemma~\ref{lem:ptasruntime}. Further in Lemma~\ref{lem:errorbound} we showed the correctness and
that the solution computed by the DP is actually a $(1+O(\eps))$-approximation of the solution.
\end{proof}

\section{Missing proofs and details from Section~\ref{sec:general-case}}

\subsection{Proof of Proposition~\ref{prop:trivhourglass}}
\trivhourglass*
\begin{proof}
If $k=2$ then there is no value $i\in \{2, \dots, k-1\}$ and hence the hourglass property is
trivially satisfied.
\end{proof}

\subsection{Proof of Lemma~\ref{lem:vcreduction}}
\vcreduction*
\begin{proof}

We first show that if there is a vertex cover of size $\gamma$ then there is a solution of cost $\gamma$ solution to our instance of \kstabbing{}.
Given
a solution $S=\{v_{1},v_{2},\ldots,v_{\gamma}\}$ to the vertex cover instance,
construct a solution to the stabbing instance as follows: for each
$v_{i}\in S$, stab the corresponding $s_{i}$ along its top edge
by a segment of length one. Clearly the cost of this set of segments
is $\gamma$. Now we notice that every \kshape{} $r_{i,j}$ corresponds
to an edge $e(v_{i},v_{j})$ in the graph. Since this edge has been
covered by one of its adjacent vertices $v_{i}\in S$, $r_{i,j}$
is also stabbed by the segment that stabs $s_{i}$. We know that every
edge of the graph is covered by some vertex in $S$, and hence every
\kshape{} in the instance is also stabbed in the solution we constructed.

Next, we argue that a solution of cost $\gamma$ to our instance of
yields a solution to vertex cover of size at most $\gamma$. Consider
any solution to the stabbing instance of cost $\gamma$. We can assume
that there are no segments of length greater than one in this solution,
since any segment of length at least $n+1$, can be broken down into
at most $n$ segments of length $1$ stabbing the same set of \kshapes{},
but along their bordering squares; and segments of length in the range
$(1,n+1)$ can stab only one \kshape{}, and hence be shortened to
length one. Further segments in any solution can also not be of length
less than one, since such a segment cannot stab any \kshape{}. Hence
we conclude that all segments in the solution are of length one, and
by extension that they stab any \kshape{} along one of its bordering
squares.

Now we construct a vertex cover solution by picking the vertices $v_{i}$,
that correspond to any square $s_{i}$ that has been stabbed by the
given (or modified as mentioned above) \kstabbing{} solution. Note
that every \kshape{} is stabbed by the given solution, and hence
a vertex adjacent to every edge in the vertex cover instance has been
picked by us. This shows that the selected set is in fact a valid
vertex set, and is of size at most $\gamma$. \qedhere
\end{proof}

\subsection{Proof of Lemma~\ref{lem:screduction}}
\screduction*
\begin{proof}
%
Given
a solution $H=\{v_{1},v_{2},\ldots,v_{\gamma}\}$ to the hitting set instance,
construct a solution to the stabbing instance as follows: for each
$v_{i}\in H$, stab the corresponding $s_{i}$ along its top edge
by a segment of length one. Clearly the cost of this set of segments
is $\gamma$. Now we notice that every \kshape{} $r_{S_{i}}$ corresponds
to an set $S_{i}\in\calF$ . Since this set has been hit by a vertex,
say $v_{i}\in H$, $r_{S_{i}}$ is also stabbed by the segment that
stabs $s_{i}$. We know that every set of the family $\calF$ is covered
by some element in $H$, and hence every \kshape{} in the instance
is also stabbed in the solution we constructed.

Similar
to the argument in proof of Lemma~\ref{lem:vcreduction}, we can
modify the solution without increasing its cost so that every segment
is of length one.
Now we construct a hitting set solution by picking the vertices $v_{i}$,
that correspond to any square $s_{i}$ that has been stabbed by the
given (or modified as mentioned above) \kstabbing{} solution. Note
that every \kshape{} is stabbed by the given solution, and hence
an element from each set in $\calF$ has been picked by us. This shows
that the selected set is in fact a valid hitting set, and is of size
at most $\gamma$. \qedhere
\end{proof}

\subsection{Proof of Lemma~\ref{lemma:candidates}}
\candidates*

\begin{proof} Let $\SOPT$ be the set of line segments of an optimal
    solution to $\calK$. Any segment $s=[x^{\ell},x^{r}]\times\{y\}\in\SOPT$
must start and end on vertical boundaries of \kshapes{}. Furthermore,
$s$ can be translated along the $y$-direction to the nearest horizontal
boundary of some part of a \kshape{} without changing the set of
\kshapes{} it stabs. Therefore, it is sufficient to consider $O(kn)$
choices each for $x^{\ell}$, $x^{r}$, and $y$, obtaining $O((kn)^{3})$
combinations.

Candidate segments not stabbing any \kshapes{} can be discarded.
Suppose that there are multiple candidate segments stabbing the same
set of \kshapes{}. Any optimal solution uses at most one of them,
specifically a shortest one. Discard the rest. This yields a set $\calC$
with no redundant segments. \end{proof}

\end{document}